\documentclass[copyright]{eptcs}
\usepackage{amsmath}
\usepackage{amsthm}
\usepackage{ tipa }
\usepackage{semantic}
\usepackage{mathtools}
\usepackage{tikz-qtree}
\usepackage{ dsfont }
\usepackage{environ}
\usepackage{booktabs}
\usepackage{nicefrac}
\usepackage{xifthen}
\usepackage{listings}
\usepackage{float}
\usepackage{ upgreek }
\usepackage{changepage}
\usepackage{ stmaryrd }
\usepackage{amssymb}
\usepackage{ mathrsfs }
\usepackage{eucal}
\usepackage[toc,page]{appendix}

\renewcommand*{\models}{%
  {\mkern5mu\|\mkern-9mu=}
}

\newcommand{\rt}[2][]{%
\ifthenelse{\isempty{#1}}{\textsc{(#2)}}{($\textsc{#2}_\textsc{#1}$)}
}
\newcommand{\eqdef}{\overset{\underset{\mathrm{def}}{}}{=}}
\newcommand{\nil}{\mathbf{0}}

\renewcommand{\o}[1]{\overline{#1}}
\renewcommand{\u}[1]{\underline{#1}}
\renewcommand{\l}{\langle}
\renewcommand{\r}{\rangle}
\newcommand{\ite}[3]{\texttt{if }#1\texttt{ then }#2\texttt{ else }#3}
\renewcommand{\>}{\triangleright}
\renewcommand{\<}{\triangleleft}

\def\authorrunning{J. Aagaard, H. H\"uttel, M. Jakobsen and M. Kettunen}
\def\copyrightholders{J. Aagaard, H. H\"uttel,\\ M. Jakobsen and M. Kettunen}
\def\titlerunning{Context-Free Session Types for Applied Pi-Calculus}

\begingroup
\catcode`\|=\active
\gdef|{\text{ \textpipe\ }}
\gdef\|{\text{\textpipe}}
\endgroup
\AtBeginDocument{\mathcode`|="8000 }

\NewEnviron{abstractsyntax}{%
  \vspace{-1.5em}
  \begin{align}
  \begin{split}
    \BODY
  \end{split}
  \end{align}
}

\NewEnviron{abstractsyntax*}{%
  \vspace{-1.5em}
  \begin{align*}
    \BODY
  \end{align*}
}

\newcommand{\rulestable}[4][]{
    \bgroup
    \newcommand{\tableCols}{ll}
    \ifthenelse{\isempty{#1}}{
        \renewcommand{\tableCols}{ll}
    }{
        \renewcommand{\tableCols}{#1}
    }
    \def\arraystretch{2.5}
    \begin{table}[htb]
        \centering
        \begin{tabular}{\tableCols}
            \hline
            #4\vspace{3mm}
            \\ \hline
        \end{tabular}
        \vspace{.1cm}
        \caption{#2}
        \label{#3}
    \end{table}
    \egroup
}

\newcommand{\substitute}[2]{\{\nicefrac{#2}{#1}\}}

\newcommand{\trans}[1]{\xrightarrow{#1}}
\newcommand{\nottrans}[1]{\not\trans{#1}}

\lstdefinelanguage{pi}{
    mathescape,
    basicstyle=\sffamily,
    keywordstyle=\tt,
    morekeywords={if, then, else, let, =, (, ), in},
    morecomment=[l]{\#},
    captionpos=below
} %

\newcommand{\arrowcdot}{\xleftrightarrow{\mkern-3mu\cdot}}
\newcommand{\vddash}{\models}
\newcommand{\para}{\mid}
\newcommand{\un}{\textsf{un }}
\newcommand{\lin}{\textsf{lin }}
\newcommand{\vmin}{\vdash_{\text{min}}}
\newcommand{\vnotmin}{\nvdash_{\text{min}}}
\newcommand{\dom}[1]{\text{dom}(#1)}

\newtheorem{lemma}{Lemma}
\theoremstyle{remark}
\newtheorem{example}{Example}
\theoremstyle{definition}
\newtheorem{definition}{Definition}
\newtheorem{theorem}{Theorem}

\newcommand{\aset}[1]{\{ #1 \}}

\newcommand{\besk}[1]{\langle #1 \rangle}
\newcommand{\inp}[2]{#1(#2).}
\newcommand{\outp}[2]{#1\besk{#2}.}
\newcommand{\slct}[2]{#1 \lhd #2.}
\newcommand{\brch}[2]{#1 \rhd \aset{ #2 }}
\newcommand{\new}[1]{(\nu\, #1)}

\newcommand{\fn}[1]{\text{fn}(#1)}

\newcommand{\fv}[1]{\text{fv}(#1)}
\newcommand{\bv}[1]{\text{bv}(#1)}

\newcommand{\ra}{\rightarrow}

\providecommand{\event}{EXPRESS/SOS 2018}

\author{
     Jens Aagaard\thanks{jbaa14@student.aau.dk} \qquad Hans
     H\"uttel\thanks{hans@cs.aau.dk} \qquad Mathias Jakobsen
     \thanks{msja15@student.aau.dk} \qquad Mikkel Kettunen \thanks{mkettu16@student.aau.dk}
    \institute{Department of Computer Science, Aalborg University, Aalborg, Denmark}
}

\title{Context-Free Session Types for Applied Pi-Calculus}

\begin{document}
\maketitle

\begin{abstract}
We present a binary session type system using context-free session types to  a version of the applied pi-calculus of Abadi et. al. where only base terms, constants and channels can be sent. Session types resemble process terms from BPA and we use a version of bisimulation equivalence to characterize type equivalence. We present a quotiented type system defined on type equivalence classes for which type equivalence is built into the type system. Both type systems satisfy general soundness properties; this is established by an appeal to a generic session type system for psi-calculi.

\end{abstract}
\section{Introduction}

Binary session types \cite{Honda1998} describe the protocol followed
by the two ends of a communication medium, in which messages are
passed. A sound type  system of this kind guarantees that a well-typed
process does not exhibit communication errors at runtime. Session
types have traditionally been used to describe linear interaction
between partners \cite{Vasco2011}, but later type systems are able to
distinguish between linear and unlimited channel usages. In particular
Vasconcelos has proposed session types with $\textsf{lin}/\textsf{un}$
qualifiers that describe linear interaction as well as shared
resources \cite{Vasco2011}.  \emph{Context-free session types} introduced by
\cite{ThiemannVasco2016} are more descriptive than the regular types
described in previous type systems \cite{Honda1998,Vasco2011} in that
they allow full sequential composition of types. Because of this,
session types can now describe protocols that cannot be captured in
the regular session types, such as transmitting complex composite data
structures. 

Many binary session type systems are concerned with languages that use
the selection and branching constructs introduced in \cite{Honda1998}
in addition to the normal input and output constructs. These
constructs are synchronisation operations, where two processes
synchronise on a channel, and are similar to method invocations found
in object-oriented programming. A branching process $\brch{l}{l_1 :
  P_1, \ldots, l_n : P_n}$ continues as process $P_k$ together with $P$ if label $l_k$
is selected on channel $c$ using the selection $\slct{c}{l_k}P$.

In this article we consider a session type system for a version of the
applied pi-calculus, due to Abadi et. al. \cite{Abadi2016}. This is an
extension of the pi-calculus \cite{DBLP:books/daglib/0098267} with
terms and extended processes and in our case extended further with selection and
branching. Our version is a ``low-level'' version in that allows one
to build composite terms but only allows for the communication of
names and nullary function symbols. In this way, the resulting version
is close to the versions of the pi-calculus used for encoding
composite terms \cite{DBLP:books/daglib/0098267}.

Our session type system combine ideas from the type systems from
\cite{ThiemannVasco2016} and \cite{Vasco2011} into a type system in
which session types are context-free and use
$\textsf{lin}/\textsf{un}$ qualifiers in order to distinguish between
linear and unbounded resources. The resulting type system uses types
that are essentially process terms from a variant of the BPA process
calculus \cite{Bergstra1988}.

In Section \ref{sec:syntax} we present our version of the applied
pi-calculus and in Section \ref{sec:typesystem} we define the
syntax and semantics of our session type system. We then prove the
soundess of our type system by using the general results about
psi-calculi from \cite{Huttel2016}. This is done by showing that the
applied pi-calculus is a psi-calculus, and that our type system is
a instance of the generic type system presented by H\"uttel in
\cite{Huttel2016}. 

Lastly we present type equivalence between types by introducing a
notion of type bisimulation for endpoint types. By considering
equivalence classes under type bisimilarity we get a new quotiented
type system whose types are equivalence classess. It then follows from
the theorems in \cite{Huttel2016} that the general results for our
first type system also hold for the quotiented type system.

\section{Applied pi-Calculus}
\label{sec:syntax}

We consider a ``low-level'' version of the \emph{applied
  pi-calculus} \cite{Abadi2016} in which composite terms are
allowed but only the transmission of simple data is possible: Only
names $n$, constants represented by functions $f_0$ with arity $0$ and
functions that evaluate to values of base types can be transmitted. We
use the notation $\widetilde{M}$ to represent the sequence
$M_1,\ldots,M_i$ and $\widetilde{x}$ to represent the sequence of
variables $x_1,\ldots,x_i$. We always assume that our processes are
specified relative to a family of parameterized agent definitions that
are on the form $N(\widetilde{x})\eqdef A$ and that every agent
variable $N$ occurring in a process has a corresponding definition.

The formation rules for processes $P$ and extended processes $A$ can
be seen below in \eqref{gram:processes}. Note that we distinguish
between variables ranged over by $x, y \ldots$ and names ranged over
by $m,n, \ldots$. We let $a$ range over the union of these sets. We
extend the syntax of processes with branching
$\brch{c}{l_1:P_1,\ldots,l_k:P_k} $ and selection $\slct{c}{l}P$ where
$l_1, \ldots l_k$ are taken from a set of labels.

Extended processes extend processes with the ability to use
\emph{active substitutions} of the form $\substitute{x}{M}$ that
instantiate variables.


\begin{abstractsyntax}
    \label{gram:processes}
    P ::=\ &\nil \mid P_1 \para P_2 \mid !P \mid \new{n}P \mid \new{n}
    P \mid \ite{M_1=M_2}{P_1}{P_2} \\
    \mid & \, \inp{n}{x}P \mid \outp{n}{u}P \mid \slct{c}{l}P \mid c
    \rhd \{ l_1:P_1,\ldots,l_k:P_k \}  \mid N(\widetilde{M})\\
    u ::= \ & n \mid x \mid f_0 \\
    M ::=\ & n \mid x \mid f(\widetilde{M}) \\
    A ::=\ &P \mid A_1 \para A_2 \mid (\nu n)A \mid \new{x}A \mid \substitute{x}{M}\\
\end{abstractsyntax}

The notion of structural congruence extends that of the usual
pi-calculus \cite{DBLP:books/daglib/0098267}. The following two
further axioms that are particular to the
applied pi-calculus are of particular importance, as they show the
role played by active substitutions.
\begin{align*}  P \para \substitute{x}{M} \equiv P\substitute{x}{M}
  \para \substitute{x}{M} & & & \nil \equiv \new{x} \substitute{x}{M}
                            \end{align*}
Together with the axioms of \cite{DBLP:books/daglib/0098267} they
allow us to factor out composite terms such that they only occur in 
active substitutions. For instance, we have that $\ite{M_1 = M_2}{P}{Q}
\equiv \nu{x}\nu{y}(\ite{x=y}{P}{Q} \para \substitute{x}{M_1} \para \substitute{y}{M_2})$. We can
therefore use a term by only mentioning the variable associated with
it. In our version of the applied pi-calculus we very directly make use of this.

Our semantics consists of reduction semantics and a labelled operational semantics that
extend those of \cite{Abadi2016} with rules for branching and
selection. Reductions are of the form $P \ra P'$, while transitions are of the form $P \trans{\alpha} P'$ where the
label $\alpha$ is given by
\[ \alpha ::= c \lhd l \mid c \rhd l \mid a(x) \mid \overline{a}x \mid
  \tau \]
and $c \lhd l$ is the co-label of $c \rhd l$. The rules defining
reductions and labelled transitions are found in Tables
\ref{tab:red-rules} and \ref{tab:trans-rules}, respectively.

\rulestable{Reduction rules}{tab:red-rules}{
    \rt{Com} & $\outp{x}{y}P \para \inp{x}{y} Q \ra P \para Q$ \\[3mm]
    \rt{Select} & $\slct{c}{l_k}P \para \brch{c}{l_1 : P_1,
                  \ldots, l_n : P_n} \ra P \para P_k$ \quad for $1 \leq k
                  \leq n$ \\[3mm]
    \rt{Match-True} & $\ite{M=M}{P}{Q} \ra P$ \\[3mm]
    \rt{Match-False} & $\ite{M=N}{P}{Q} \ra Q$
 }

 \begin{table}[h]
   \begin{center}
     \hrule
     \begin{tabular}{llll}
       & & \\
      \rt{Red} & \inference{P \ra P'}{P \trans{\tau} P'} & 
     \rt{Select} & $\slct{c}{l}P \trans{c \lhd l}
                   P$ \\[4mm]
     \rt{Branch} & $\brch{c}{l_1:P_1,\ldots,l_k:P_k} \trans{ c \rhd
                                                  l_i} P_i$ &
                       \rt{Par} & \inference{P \trans{\alpha} P'}{P \para Q
                                  \trans{\alpha} P' \para Q}  \\[5mm]
               & for $1 \leq i \leq k$ & &  if
               $\bv{\alpha} \cap \fv{Q} = \emptyset$\\[4mm]
  \rt{Output} & $\outp{x}{y}P \trans{x\besk{y}} P$ & \rt{Input} & $\inp{x}{y}P \trans{x(y)} P$ \\[3mm]
 
  \rt{New} & \inference{P \trans{\alpha} P'}{\new{n}P \trans{\alpha}
             \new{n}P'} \quad if $n \notin \fn{\alpha}$ & 
  \rt{Struct} & \inference{P \equiv Q \quad Q \trans{\alpha} Q' \quad Q'
                \equiv P'}{P \trans{\alpha} P} \\
       & & \\
    \end{tabular}
    \hrule
\end{center}
\caption{Labelled transition rules}
\label{tab:trans-rules}
\end{table}


\section{A context-free session type system}
\label{sec:typesystem}
We now present the syntax and semantics of our type system.

\subsection{Session types}
\label{sec:sessiontypeintroduction}

Session types describe the communication protocol followed by a
channel. Consider writing a process that transmits a binary tree where
each internal node containsan integer. We want to transmit this tree
by sending only base types (i.e. the integers) on a channel. The tree
data type can be described with the grammar  below. 
\begin{equation*}
    \textsf{Tree} ::= (\textsf{Int}, \textsf{Tree}, \textsf{Tree})\ |\ \textsf{Leaf}
\end{equation*}
Using the regular types introduced in \cite{Vasco2011}, the type for a
channel transmitting such a data structure could be the recursive type
$\mu z.\oplus\{\textsf{Leaf}: \lin \textsf{skip}\quad \textsf{Node}:
\lin !\textsf{Int}.z\}$. In this type $z$ is a type variable that is
recursively defined. The type describes how a single node is
transmitted, if it is a leaf node we do not do anything, if it is an
internal node then the integer value is transmitted with the output
type $!\textsf{Int}$ and then the sub-trees of the node are
transmitted with a recursive call.

However, if we use this regular session type, we are not able to
guarantee that the tree structure is preserved. The reason is that the
session type describes that a list of nodes are being sent, but not
the position in the tree of each node. On the other hand, if we use
the context-free session type disciple introduced by Thiemann and
Vasconcelos \cite{ThiemannVasco2016}, we can specify the preservation
of tree structures by using types such as
$\mu z.\oplus\{\textsf{Leaf}: \textsf{skip}\quad \textsf{Node}:
!\textsf{Int};z;z\}$. Using sequential composition with the $\_;\_$
operator, we can specify a protocol that will guarantee that the tree
structure by first sending the left sub-tree and then the right
sub-tree. Introducing a sequential operator can introduce challenges
for typing a calculus, as the following example shows. If we
were to reuse a channel by sending an integer after transmitting a
tree, the type would be
$\mu z.\oplus\{\textsf{Leaf}: \textsf{skip}\quad \textsf{Node}:
!\textsf{Int};z;z\};!\textsf{Int}$.
\begin{equation}
    \label{eq:selectRule} 
    \inference{\Gamma(c)=\oplus\{l_i: T_{E_i}\}_{i\in I}}{\Gamma \vdash \text{select } l_i \text{ on } c}
\end{equation}
When typing rules are created for a calculus, it is often defined on
the structure of types and terms. An example of such a typing rule for
a select statement is shown in \eqref{eq:selectRule}, which says that
a select statement is well typed if the select operation is performed
on a channel with a select type. If however the channel has a type as
shown before, the select rule cannot be used, as the channel has the
type of a sequential composition, and inside the sequential
composition we have a recursive type.

We require an equirecursive treatment of types, which allows us to
expand the type to $\oplus\{\textsf{Leaf}: \textsf{skip}\quad
\textsf{Node}: !\textsf{Int};\mu z.\oplus\{\textsf{Leaf}:
\textsf{skip}\quad \textsf{Node}: !\textsf{Int};z;z\};\mu
z.\oplus\{\textsf{Leaf}: \textsf{skip}\quad
\textsf{Node}: !\textsf{Int};z;z\}\};!\textsf{Int}$. This is achieved
by unfolding a recursive type $\mu z.T$ to $T$ where all occurrences
of $z$ in $T$ are replaced with the original $\mu z.T$. So now we are
left with a sequential composition with a select type and output type.
In order to transform this into a select type, we need a distributive
law that allows us to move the sequential composition inside the
select type, in order to obtain the type $\oplus\{\textsf{Leaf}:
\textsf{skip};!\textsf{Int}\quad \textsf{Node}: !\textsf{Int};\mu
z.\oplus\{\textsf{Leaf}: \textsf{skip}\quad
\textsf{Node}: !\textsf{Int};z;z\};\mu z.\oplus\{\textsf{Leaf}:
\textsf{skip}\quad
\textsf{Node}: !\textsf{Int};z;z\};!\textsf{Int}\}$. Such a rule does
exist, we will introduce a method to prove that these types exhibit
the same behaviour in Section \ref{sec:typeequivalence} when we
introduce type equivalence. 

\subsection{The language of types}

We denote the set of all types by $\mathcal{T}$; the types $T\in
\mathcal{T}$ are described by the formation rules in
\eqref{gram:types}. These session types are a modified version of the
context-free session types presented by \cite{ThiemannVasco2016}. The
modifications made to the types are that we allow input and output
session types to transmit other session types to allow sending
channels in other channels, and finally that we introduce the \lin and
\un qualifiers. We let $B$ range over a set of base types.

\begin{abstractsyntax}
    \label{gram:types}
    p ::=\ &\textsf{skip} \mid ?T \mid !T \mid \&\{l_i:T_{E_i}\} \mid \oplus\{l_i:T_{E_i}\}  \\
    q ::=\ &\lin \mid \un \\
    T_E ::=\ &q\ p \mid z \mid \mu z.T_E \mid T_{E_1}; T_{E_2} \\
    T ::=\ & S \mid B \mid T_E \\
    S ::=\ & (T_{E_1}, T_{E_2})
\end{abstractsyntax}

From the formation rules we can see that a session type $S$ is a pair
of endpoint types $T_{E_1}$ and $T_{E_2}$. Endpoint types describe one
end of channel, and are the types that evolve when a channel is used.
The qualifiers \lin and \un from \cite{Vasco2011} are used to describe
a linear interaction between two partners and a unrestricted shared
resources respectively. An example of an \un type could be a server,
that a lot of processes have access to. To ensure that no
communication errors occur if multiple processes can read from a
channel concurrently, we require that the type must never change
behaviour. This means that an \un type must be the same before and
after a transition. 

\subsection{Transitions for types}\label{sec:typetransitions}

We now present the transition rules for endpoint types, which describe
how the types can evolve when an action is performed. We use an
annotated reduction semantics to describe the behaviour of our types.
Our labels are generated by the grammar in \eqref{gram:actions} where
we have select, branch, input and output actions. The transitions are
of the form $T_E\trans{\lambda}T_E'$ and are shown in Table
\ref{tab:typetransitions}. We let $T_E\substitute{x}{y}$ be the
endpoint type $T_E$ where all free occurrences of $x$ has been
replaced with $y$.  

\begin{equation}
  \lambda ::= !T_1 \mid ?T_1 \mid \rhd \, l \mid \lhd \, l
  \label{gram:actions}
\end{equation}

In the transition rules, we use the function $Q$ defined in
\eqref{eq:q} to find the qualifier of compound types such as recursive
types or sequential composition types.   
\begin{align}
\label{eq:q}
\begin{split}
    {Q}(q\ p) &\eqdef q \\
    {Q}(T_{E_1};T_{E_2}) &\eqdef {Q}(T_{E_1}) \\
    {Q}(\mu z.T_E) &\eqdef {Q}(T_E)
\end{split}
\end{align}
A relation $\sqsubseteq=\{(\textsf{lin},\textsf{un}),(\textsf{un},\textsf{un}),(\textsf{lin},\textsf{lin})\}$
is defined for qualifiers in \cite{Vasco2011}. We also follow the
definition of $q(T)$ and $q(\Gamma)$ from that of \cite{Vasco2011}. In
short $\textsf{lin}(\Gamma)$ is always satisfied, and
$\textsf{un}(\Gamma)$ is satisfied iff all elements in $\Gamma$ are
unrestricted. 

The type system contains the sequential operator $\_;\_$ as well as
choice operators; select $\&\{\ldots\}$ and branch $\oplus\{\ldots\}$.
This is very similar to Basic Process Algebra (BPA)\cite{Bergstra1988}
that contains the sequential operator $\cdot$ and nondeterministic
choice operator $+$. We also have recursive types, which corresponds
to variables with recursive definitions in BPA. A BPA expression is
\textit{guarded} if all recursive variables on the right hand side is
preceded by an action $\lambda$ \cite[p. 53]{Bergstra1988}. Similarly
we say that a type is guarded if every recursion variable is preceded
by an input or output. These similarities with BPA will become very
important, as we can describe types as BPA expressions and by showing
results about these expression, we can in turn show results about our
types. 

\rulestable[llll]{Annotated reduction semantics for types}{tab:typetransitions}{
    \rt{Input} & \inference{}{q\ ?T\trans{?T}q\ \textsf{skip}} &
    \rt{Output} & \inference{}{q\ !T\trans{!T}q\ \textsf{skip}} \vspace{0.1cm}\\[5mm]
    \rt{Seq1} & \inference{T_{E_1} \trans{\lambda} T'_{E_1} & Q(T_{E_1})\sqsubseteq Q(T_{E_1}')}{T_{E_1} ; T_{E_2} \trans{\lambda}T'_{E_1} ; T_{E_2}} &
    \rt{Seq2} & \inference{T_{E_1} \nottrans{} & Q(T_{E_1})\sqsubseteq
      Q(T_{E_2}) \\ \\ T_{E_2} \trans{\lambda} T'_{E_2} & Q(T_{E_2})\sqsubseteq Q(T_{E_2}')}{T_{E_1};T_{E_2} \trans{\lambda} T'_{E_2}} \\
    \rt{Select} & \inference{q \sqsubseteq Q(T_{E_k})}{q\ \oplus\{l_i : T_{E_i}\}_{i\in I} \trans{\triangleleft l_k} T_{E_k} } &
    \rt{Branch} & \inference{q \sqsubseteq Q(T_{E_k})}{q\ \&\{l_i : T_{E_i}\}_{i\in I} \trans{\triangleright l_k} T_{E_k} }\\
    \rt{Rec} & \inference{T_E\substitute{z}{\mu z.T_E}\trans{\lambda}T'_E}{ \mu z .T_E\trans{\lambda}T_{E_1}'}
}

\subsection{Typing rules}

We now present a type system for our version of the applied
pi-calculus in Table \ref{tab:typesystem}. In this type system, the
type judgements for processes are on the form $\Gamma \vdash P$
meaning that the process $P$ is well typed in context $\Gamma$. The
judgments for terms are on the form $\Gamma\vdash M:T$ meaning that
the term $M$ is well typed with type $T$ in context $\Gamma$. Lastly
the judgments for extended processes are on the form
$\Gamma \vdash_A A$ meaning that the process $A$ is well typed in
context $\Gamma$. We type processes in a type context $\Gamma$ which
contains types for the variables, names and functions symbols of a
process. We follow the definition of a type context from
\cite{Vasco2011}: $\emptyset$ is the empty context, $\Gamma, x: T$ is
the context equal to $\Gamma$ except that $x$ has the type $T$ in the
new context. This operation is only defined when
$x\notin \text{dom}(\Gamma)$.

Table \ref{tab:typesystem} shows the typing rules for processes. We use the context split $\circ$ and context update $+$ operations from \cite{Vasco2011}. 

The context split operation is used to split a context into two
constituents. A maximum of two processes must have access to a given
linear session type; a context either contains the entire session type
$S=(T_{E_1},T_{E_2})$, or a single endpoint type $T_{E}$. When
splitting a context into two, we can pass $S$ to either context, or
one endpoint to each context. If the context only has an endpoint
type, the endpoint type can only be passed on to one of the two
contexts. This way we ensure that each \lin endpoint of a channel is
known in exactly one context. Names of unrestricted type can be shared
among all contexts. 

The context update operation updates the type of a channel. The
$\Gamma, x:T$ operation is only defined when
$x\notin\text{dom}(\Gamma)$. The $+$ operation
$\Gamma=\Gamma_1+\Gamma_2$ uses the type of $x$ in $\Gamma_2$ to
update the type in $\Gamma_1$. So if $\Gamma_1(x)=T_1$ and
$\Gamma_2(x)=T_2$ then $\Gamma(x)=T_2$. The two operations are used in
\rt{Input} and \rt{Output} where we must split our context into two,
to type each endpoint of a linear channel in its own context.

\rulestable{Typing rules for Extended Processes}{tab:extendedProcesses}{
    \rt{Plain} & \inference{\Gamma \vdash P}{\Gamma \vdash_A P} \\
    \rt{Par} & \inference{\Gamma_1 \vdash_A A_1 & \Gamma_2 \vdash_A A_2 }{\Gamma_1 \circ \Gamma_2 \vdash_A A_1 | A_2} \\
    \rt{Name-Res} & \inference{\Gamma, n: S \vdash_A A }{\Gamma \vdash_A (\nu n)A} \\
    \rt{Var-Res} & \inference{\Gamma, x: T \vdash_A A }{\Gamma \vdash_A (\nu x)A} \\
    \rt{Sub} & \inference{\Gamma \vdash x: T & \Gamma \vdash M: T}{\Gamma \vdash_A \substitute{x}{M}} 
  }
  
\rulestable[c]{Context split for Applied $\pi$-calculus, based on  \cite{Vasco2011}}{tab:contextsplit}{
    \inference{}{\emptyset = \emptyset \circ \emptyset} \\
    \inference{\Gamma_1 \circ \Gamma_2 = \Gamma & Q(T_E) = \un}{\Gamma, n: T_E = (\Gamma_1, n:T_E) \circ (\Gamma_2, n:T_E)} \\
    \inference{\Gamma_1 \circ \Gamma_2 = \Gamma & S = (T_{E_1}, T_{E_2}) & Q(T_{E_1}) = \un & Q(T_{E_2}) = \un}{\Gamma, n: S = (\Gamma_1, n:S) \circ (\Gamma_2, n:S)} \\
    \inference{\Gamma_1 \circ \Gamma_2 = \Gamma & Q(T_E) = \lin}{\Gamma, n: T_E = (\Gamma_1, n:T_E) \circ \Gamma_2} \\
    \inference{\Gamma_1 \circ \Gamma_2 = \Gamma & Q(T_E) = \lin}{\Gamma, n: T_E = \Gamma_1 \circ (\Gamma_2, n:T_E)} \\
    \inference{\Gamma_1 \circ \Gamma_2 = \Gamma & S = (T_{E_1}, T_{E_2}) & Q(T_{E_1}) = \lin & Q(T_{E_2}) = \lin }{\Gamma, n: S = (\Gamma_1, n:S) \circ \Gamma_2} \\
    \inference{\Gamma_1 \circ \Gamma_2 = \Gamma & S = (T_{E_1}, T_{E_2})& Q(T_{E_1}) = \lin & Q(T_{E_2}) = \lin }{\Gamma, n: S = \Gamma_1 \circ (\Gamma_2, n:S)} \\
    \inference{\Gamma_1 \circ \Gamma_2 = \Gamma & S = (T_{E_1}, T_{E_2})& Q(T_{E_1}) = \lin & Q(T_{E_2}) = \lin }{\Gamma, n: S = (\Gamma_1, n: T_{E_1}) \circ (\Gamma_2, n: T_{E_2})} \\
    \inference{\Gamma_1 \circ \Gamma_2 = \Gamma & S = (T_{E_1}, T_{E_2})& Q(T_{E_1}) = \lin & Q(T_{E_2}) = \lin}{\Gamma, n: S = (\Gamma_1, n: T_{E_2}) \circ (\Gamma_2, n: T_{E_1})}
}

\rulestable[ll]{Context update for Applied $\pi$-calculus from \cite{Vasco2011}}{tab:contextupdate}{
    \inference{}{\Gamma = \Gamma + \emptyset} & 
    \inference{\Gamma = \Gamma_1 + \Gamma_2}{\Gamma, n: T = \Gamma_1 + (\Gamma_2, n :T)} \\
    \multicolumn{2}{c}{
    \inference{\Gamma = \Gamma_1 + \Gamma_2 & Q(T_E) = \un}{\Gamma, n: T_E = (\Gamma_1, n: T_E) + (\Gamma_2, n: T_E)}}
}

\rulestable{Typing rules for terms}{tab:termsrules}{
    \rt{Name} & \inference{\textsf{un}(\Gamma)}{\Gamma, n:T \vdash n:T} \\
    \rt{Variable} & \inference{\textsf{un}(\Gamma)}{\Gamma, x:T \vdash x:T} \\
    \rt{Fun} & \inference{\Gamma(f) = T_1 \times T_2 \times \dots \times T_k -> T}{\Gamma \vdash f:T_1 \times T_2 \times \dots \times T_k -> T}
}

\rulestable{Typing rules for processes}{tab:typesystem}{
    \rt{Nil} &\inference{\textsf{un}(\Gamma)}{\Gamma \vdash \nil} \\
    \rt{Par} &\inference{\Gamma_1 \vdash P_1 & \Gamma_2 \vdash P_2}{\Gamma_1 \circ \Gamma_2 \vdash P_1 | P_2} \\
    \rt{Repl} & \inference{\Gamma \vdash P & \textsf{un}(\Gamma)}{\Gamma \vdash !P}  \\ 
    \rt{Res} & \inference{\Gamma,n:S \vdash P}{\Gamma \vdash (\nu n)P} \\
    \rt{If} & \inference{\Gamma \vdash P_1 & \Gamma \vdash P_2 & \Gamma \vdash M_1 : T & \Gamma \vdash M_2 : T}{\Gamma \vdash \ite{M_1=M_2}{P_1}{P_2} } \\
    \rt{Input} & \inference{\Gamma_1 \vdash n:T_E & T_E \trans{?T} T_E' & \Gamma_2, x:T + n:T_E' \vdash P}{\Gamma_1 \circ \Gamma_2 \vdash n(x).P } \\
    \rt{Output} & \inference{\Gamma_1 \vdash n:T_E & T_E \trans{!T} T_E' & \Gamma_2 \vdash M:T & \Gamma_2 + n:T_E' \vdash P}{\Gamma_1 \circ \Gamma_2 \vdash \o{n} \l M \r .P } \\
    \rt{Select} & \inference{T_E\trans{\<l}T_E' & \Gamma, c:T_E'\vdash P}{\Gamma, c:T_E \vdash c\< l.P} \\
    \rt{Branch} & \inference{T_E\trans{\> l_i}T_E' & \Gamma, c:T_E'\vdash P_i }{\Gamma, c:T_E \vdash c\>\{l_1:P_1,\ldots, l_k:P_k\}} for all $1\leq i \leq k$ \\
    \rt{Agent} & \inference{\Gamma \vdash_A A|\substitute{\widetilde{x}}{\widetilde{M}}}{\Gamma \vdash N(\widetilde{M})} where $N(\widetilde{x})\eqdef A$
}

\subsection{Duality of types}
\label{sec:duality}
The notion of type duality is central to session type systems; it
expresses that the protocols followed by the two endpoints of a name
must be opposites: If one end transmits a value, the other end must
receive it. We denote the dual of an endpoint type $\o{T_E}$ and
define duality below, following \cite{ThiemannVasco2016}. 

\begin{align*}
    \o{q\ \textsf{skip}}&\eqdef q\ \textsf{skip} && \o{q\ \&\{l_1: T_{E_1}, \ldots, l_k:T_{E_k}\}}\eqdef q\ \oplus\{l_1: \o{T_{E_1}}, \ldots, l_k:\o{T_{E_k}}\}\\
    \o{q\ ?T}&\eqdef q\ !T && \o{q\ \oplus\{l_1: T_{E_1}, \ldots, l_k:T_{E_k}\}}\eqdef q\ \&\{l_1: \o{T_{E_1}}, \ldots, l_k:\o{T_{E_k}}\}\\
    \o{q\ !T}&\eqdef q\ ?T &&  \o{\mu z.T_E} \eqdef \mu z.\o{T_E} \\
    \o{T_{E_1};T_{E_2}}&\eqdef \o{T_{E_1}};\o{T_{E_2}} && \o{z} \eqdef z
\end{align*}

A session type $S = (T_{E_1},T_{E_2})$ is \emph{balanced} iff its
endpoint types are dual, that is, if $\o{T_{E_1}} = T_{E_2}$. A type
context $\Gamma$ is balanced iff every session type in the range of
$\Gamma$ is balanced. We use the notation $\Gamma \vdash_\text{bal}P$
to describe that a process $P$ is well typed in a balanced context
$\Gamma$.

\section{Applied pi-calculus as a psi-calculus instance}
\label{sec:psicalculus}
Psi-calculus is a general framework for process calculi. In this section we show that the applied pi-calculus is an instance of it, and this will then be used in Section \ref{sec:prop} to obtain results about our type system.

An instance of the psi-calculus framework contains the seven elements \cite{Bengtson2011} given in Table \ref{tab:psicalculielements}.

\begin{table}[H]
    \centering
    \begin{tabular}{ll}
         \textbf{T} & Data terms \\
         \textbf{C} & Conditions \\
         \textbf{A} & Assertions \\
         $\arrowcdot : \textbf{T} \times \textbf{T} \rightarrow \textbf{C}$ & Channel Equivalence\\
         $\otimes:  \textbf{A} \times \textbf{A} \rightarrow \textbf{A}$ & Composition \\
         $\textbf{1}: \textbf{A}$ & Unit \\
         $\vdash \subseteq \textbf{A} \times \textbf{C}$ & Entailment\\ 
    \end{tabular}
    \caption{Elements of psi-calculi}
    \label{tab:psicalculielements}
\end{table}

The syntax of an instance of the psi-calculus is described by the formation rules in \eqref{gram:psicalculus} that generalize those of the pi-calculus. The input and output constructions allows to use arbitrary terms as channels, and the input construction allows for matching on a pattern $(\lambda \tilde{x})N$; the pattern variables in $\tilde{x}$ are bound to the subterms that match. The other syntactic constructs are similar to those of the pi-calculus. The case construct is a generic case of the if-construct; we generalize match conditions $M_1=M_2$ to allow any $\phi\in\textbf{C}$ as a condition. Lastly in psi-calculi we have a concept of assertions $\Uppsi\in\textbf{A}$. These generalize the notion of active substition found in the applied pi calculus.

\begin{align}
    P :=\ &\nil \mid \o{M}N.P \mid \u{M}(\lambda\widetilde{x})N.P \mid \textbf{case}\  \phi_1:P_1[]\ldots [] \phi_n: P_n \\
    |\ &(\nu a)P \mid P\|Q \mid !P \mid \llparenthesis\Uppsi\rrparenthesis \mid M \< l.P \mid M \> \{ l_1:P_1,\ldots,l_k:P_k \}
   \label{gram:psicalculus}
\end{align}

The structural operational semantics of psi-calculi has transitions of the form $\Psi \blacktriangleright P\trans{\alpha}P'$ where the labels $\alpha$ are given by the formation rules below \cite{Huttel2016}. 

\begin{abstractsyntax}
\label{gram:psiactions}
 \alpha ::= M \< l | M \> l | \o{M}(v\widetilde{a})N | \u{K}N | (v\widetilde{a})\tau @ (v\widetilde{b})(\o{M}N\u{K}) | (v\widetilde{a})\tau @ (M\<l\>N)   
\end{abstractsyntax}

The first four actions correspond to the actions we know from the applied pi-calculus: selection, branch, output and input.
The last two action are internal $\tau$-actions that correspond to internal actions that are either an input/output-exchange or a branch/select-exchange.

It follows from\cite[p. 8]{Bengtson2011} that the standard pi-calculus is an instance of the psi-calculus. Below we do the same for the instance \textbf{APi} that shows that the applied pi-calculus is also an instance of the psi-calculus. In the definition $\mathcal{N}$ is the set of all names and $\mathcal{F}$ is the set of all function names.
\begin{align*}
    \textbf{T} &\eqdef \mathcal{N} \cup \{f(M_1, \ldots M_k) | f \in \mathcal{F}, M_i \in \textbf{T}\} & & 
    \textbf{C} && \eqdef \{M=N | M,N\in \textbf{T}\} \\
    \textbf{A} &\eqdef \{1\} & & 
    \arrowcdot &&\eqdef \{((n, n), n=n) | n \in \mathcal{N}  \} \\
    \otimes &\eqdef \{((\Uppsi_1, \Uppsi_2), 1) | \Uppsi \in \textbf{A}  \} & & 
    \textbf{1} & &\eqdef 1 \\
    \vdash &\eqdef \{(1, M=M) | M\in \textbf{T}\} 
\end{align*}
\section{Properties of our type system} \label{sec:prop}

A generic binary session type system for psi-calculi was presented in \cite{Huttel2016}. The intention is that any existing binary session type systems for process calculi can be captured as special instances of the generic system, as long it satisfies four specific requirements. We already know that our applied pi-calculus is a simple psi-calculus and now establish that our type system fulfils the requirements of \cite{Huttel2016}. This will then allow us to obtain the usual results for binary type systems as a simple corollary of the theorems for the generic system.

\subsection{Transition structure}

Type transitions in our type system are an instance of the generic type transitions in \cite{Huttel2016}. Both the generic type transitions and our type transition consist of send, receive, branch and select. The syntax of type transitions is uniform across the two articles as they are both generated by the grammar in \eqref{gram:typetrans}. This illustrates that there is a one-to-one correspondence between the type transitions in the two type systems.

\begin{equation}
  \lambda ::= \< l | \> l | !T | ?T
      \label{gram:typetrans}
\end{equation}

\subsection{Revisiting duality}
Duality of types in our pi-calculus is defined on the structure of types, as seen in Section \ref{sec:duality}. In \cite{Huttel2016} duality is defined on type transitions, where \eqref{eq:typetransdual} holds for dual types (we use $\u{\,\cdot\,}$ to denote duality defined on type transitions).
\begin{equation}
    \label{eq:typetransdual}
    T_E \trans{\lambda} T_E' \Leftrightarrow \u{T_E} \trans{\o{\lambda}} \u{T_E'}
\end{equation}
In \eqref{eq:dualInOut} and \eqref{eq:dualLabels} we see the duality of type transitions, as presented by \cite{Huttel2016}.
\begin{align}
    \label{eq:dualInOut}
        \o{!T_1} = ?T_2 && \o{?T_1} = !T_2  \\
    \label{eq:dualLabels}
        \o{\<l} = \>l && \o{\>l} = \<l
\end{align}

We now show that the duality defined in Section \ref{sec:duality} upholds the property in \eqref{eq:typetransdual}.

\begin{lemma}
\label{lemma:duality}
$\u{T_E} =\o{T_E}$
\end{lemma}

\begin{proof}
By induction in the structure of types. 
\end{proof}

\subsection{Checking requirements}

In the type system presented in \cite{Huttel2016}, type judgements are
relative to a type context and an assignment and therefore of the form
$\Gamma,\Uppsi \vdash \mathscr{J}$, where the judgment body is either
$\mathscr{J}$ is either a term typing $M : T$ or $P$, the statement
that process $P$ is well-typed. We write $\Gamma,\Uppsi
\vdash_\text{min} \mathscr{J}$, if $\Gamma',\Uppsi'
\not\vdash_\text{min} \mathscr{J}$ for every smaller $\Gamma'$ and
$\Uppsi'$.  

For each requirement presented in \cite{Huttel2016} we show that it is satisfied in our type system.

\paragraph{Requirement 1:} If $\Gamma_1, \Uppsi_1 \vdash_\text{min}\mathscr{J} $ and $\Gamma_2, \Uppsi_2 \vdash_\text{min} \mathscr{J}$ then $\Gamma_1 = \Gamma_2 $ and $\Uppsi_1 \simeq \Uppsi_2$. \\

As we only have the assertion \textbf{1} in type judgements,
$\Uppsi_1\simeq\Uppsi_2$ is trivially fulfilled as both are
\textbf{1}. Let $\Gamma_1,\Gamma_2$ be type contexts such that
$\Gamma_1, \textbf{1} \vmin \mathscr{J}$ and $\Gamma_2, \textbf{1}
\vmin \mathscr{J}$. Assume that $\Gamma_1 \neq \Gamma_2$ then without
loss of generality there exists an $x$ such that $x\in\dom{\Gamma_1}$
and $x\notin\dom{\Gamma_2}$. Because judgments must be well-formed we
know that $fn(\mathscr{J})\subseteq \dom{\Gamma_1}$ and
$fn(\mathscr{J})\subseteq \dom{\Gamma_2}$, hence $x\notin
fn(\mathscr{J})$. Let $\Gamma_1'$ be defined as $\Gamma_1$ except
$x\notin\dom{\Gamma_1'}$, then $\Gamma_1',\textbf{1}\vdash
\mathscr{J}$ and $\Gamma_1'<\Gamma_1$ thus
$\Gamma_1,\textbf{1}\vnotmin\mathscr{J}$. This is a contradiction,
hence our assumption is wrong and $\Gamma_1=\Gamma_2$, which means
that the requirement is fulfilled. 

\paragraph{Requirement 2:} If $\Gamma, \Uppsi\vdash M:T@c$ then $\Gamma(c)=T_E$ for some endpoint type $T_E$. \\

In our calculus, the only terms that can be used as channels are
names. The $(\nu n)P$ construct uses a channel constructor $n$ to
declare a channel with a session type $\Gamma, \Uppsi \vdash n: S@n$.
$S$ is a session type $S=(T_{E_1}, T_{E_2})$ for some endpoint types.
When type checking $n$, only one endpoint is present in the context
$\Gamma$, in which case $\Gamma(n)=T_{E_1}$ or $\Gamma(n)=T_{E_2}$,
meaning that the requirement is satisfied.  

\paragraph{Requirement 3:}

Suppose $\widetilde{N}\in\textsc{match}(M, \widetilde{x}, X)$, $\Gamma, \Uppsi \vdash M$ and $\Gamma_1+\widetilde{x}:\widetilde{T},\Uppsi_1\vdash_{\text{min}}X:\widetilde{T}\rightarrow U$. Then there exist $\Gamma_{2i},\Uppsi_{2i}$ such that $\Gamma_{2i}, \Uppsi_{2i}\vdash_{\text{min}}N_i:T_i$ for all $1\leq i \leq \vert\widetilde{x}\vert=n$. \\

In our calculus, the only possible match is $M \in MATCH(M, x, x)$. As
$\vert\widetilde{x}\vert=1$ the requirement becomes $\Gamma_{2},
\Uppsi_{2}\vdash_{\text{min}}M : T$. From the requirement that $M$ is
well typed, this is trivially fulfilled. 


\paragraph{Requirement 4:}

If $\Uppsi \models M \arrowcdot K$ and $\Gamma, \Uppsi\vdash M:S$ then
$\Gamma, \Uppsi \vdash K:S$. If $\Uppsi \models M \arrowcdot K$ and
$\Gamma, \Uppsi \vdash M:T$ then $\Gamma, \Uppsi \vdash K:\o{T}$.

In Section \eqref{sec:psicalculus} we defined $\arrowcdot$ as $=$,
meaning that two channels are equal if it is the same name. The first
part of the requirement becomes: If $\Uppsi\vddash n = n$ and $\Gamma,
\Uppsi\vdash n:S$  then $\Gamma, \Uppsi \vdash n:S$ which is trivially
fulfilled. The second part happens when only one end of a channel is
in the context, with an endpoint type, then the other end of the
channel must have a dual endpoint type. If $\Uppsi\vddash n = c$ and
$\Gamma, \Uppsi\vdash n:T_E$ then $\Gamma, \Uppsi\vdash c:\o{T_E}$.  

A channel has a balanced session type $S=(T_E, \o{T_E})$ for some
endpoint type. If $n$ and $c$ have endpoint types and are the same
channel, then they must also have types dual of each other. If the
types are not unrestricted then they can only evolve into other types
dual of each other, due to the requirement that only two processes
have access to the channels and that if
$T_E\trans{\lambda}T_E'\Longleftrightarrow\u{T_E}\trans{\o{\lambda}}\u{T_E'}$
we know that the types will stay dual of each other. If the types are
unrestricted then the requirement that whenever
$T_E\trans{\lambda}T_E'$ for some $\lambda$ then $T_E=T_E'$, ensures
that the types will stay dual of each other.

\subsection{Fidelity result}

In this section we discuss the results that we obtain by showing that our type system is an instance of the generic type system.
H\"uttel presents and proves two main theorems for the generic type system, that we will use to show results about our type system as well. 

The first theorem, which is about well typed $\tau$-actions from \cite{Huttel2016} follows below:

\begin{theorem}[Well-typed $\tau$-actions, Theorem 9 of \cite{Huttel2016}]
  Suppose we have $\Uppsi_0 \blacktriangleright P\trans{\alpha}P'$, where $\alpha$ is a $\tau$-action and that $\Gamma,\Uppsi\vdash_{\normalfont \textsf{bal}}P$ and $\Uppsi\leq \Uppsi_0$ then for some $\Uppsi'\leq \Uppsi$ and $\Gamma'\leq\Gamma$ we have $\Gamma',\Uppsi'\vdash_{{\normalfont \text{min}}}\alpha:(T@c,U)$.
\end{theorem}

This theorem says that if a process $P$ can make a $\tau$-action and become $P'$, and that $P$ is well typed in an environment where channels have pairs of dual endpoint types, then the action is also well-typed. So internal synchronisation in a well-typed process is well typed as well.

The second theorem is about fidelity and follows below:

\begin{theorem}[Fidelity, Theorem 10 of \cite{Huttel2016}]
  Suppose we have $\Uppsi_0\blacktriangleright P \trans{\alpha}P'$, where $\alpha$ is a $\tau$-action and that $\Gamma,\Uppsi\vdash_{\normalfont\textsf{bal}}P$. Then for some $\Gamma'\leq\Gamma$ and for some $\Uppsi'\leq\Uppsi$ we have $\Gamma',\Uppsi'\vdash_{\normalfont\textsf{min}}\alpha:(T@c,U)$ and $\Gamma\pm (\alpha,(T@c,U)),\Uppsi'\vdash_{\normalfont\textsf{bal}}\circeq P'$.
  \end{theorem}

This theorem states that when an action performed in a well typed process and the action is well typed, which is guaranteed by the previous theorem, then the resulting process after the $\tau$-action is also well typed in an updated type environment.This result gives us the property that if a process is well typed, then it will never experience communication errors. The theorem also tells us that processes evolve according to the types prescription.

\section{Type equivalence}
\label{sec:typeequivalence}

In this section we discuss type equivalence in our type system and show how this leads to a new session type system. 

\subsection{Why type equivalence matters}

First we motivate type equivalence by expressing an example from \cite{ThiemannVasco2016}, in our applied pi-calculus and our type system. Recall the example from Section \ref{sec:sessiontypeintroduction} where the goal was to transmit a binary tree while preserving the tree structure.

\begin{example}
\label{ex:binarytree}

Consider the parameterised agent $S$ below. The parameter $t$ is the tree to be transmitted, $c$ is the channel the tree is sent on, and $w$ is a channel used for initiating the transition of the right sub-tree after the transmission of the left sub-tree has finished. The projection functions \textsf{fst}, \textsf{snd} and \textsf{thrd} are used to access the elements of a tuple.

\begin{lstlisting}[language=pi]
$S(t, c, w) \eqdef$
    if $t=\textsf{Leaf}$ then
        $c\< \textsf{Leaf}.\o{w}\l c \r.\nil$
    else
        $c\<\textsf{Node}.$
        $\o{c}\l\textsf{fst(t)}\r.$
        $(\nu n)(S(\textsf{snd}(t), c, n)|n(x).S(\textsf{thrd}(t), c, w))$
\end{lstlisting}

 If we analyse the $S$ process, the endpoint type of $c$ is $T_C$, which is the same type as was described in Section \ref{sec:sessiontypeintroduction}. The type describes how a single node is transmitted, if it is a leaf node we do not do anything, if it is internal node then the value is transmitted and then the sub-trees of the node.

\begin{lstlisting}[language=pi]
$T_C = \mu z.\oplus\{$
    $\textsf{Leaf}: \lin \textsf{skip} $
    $\textsf{Node}: \lin \textsf{!Int};z;z$
$\}$
\end{lstlisting}

 The process below receives the transmitted tree preserving the original structure. Through similar analysis as on the sending end, we can confirm that the endpoint type of $c$ in this process is $\o{T_C}$. 

\begin{lstlisting}[language=pi]
$R(c, w) \eqdef c\>\{$
    $\textsf{Leaf}: \o{w}\l c\r.\nil$
    $\textsf{Node}: (\nu n)(c(x).R(c, n) | n(x).R(c, w))$
$\}$
\end{lstlisting}

 We can now create a process $P$ that transmits a tree on the $c$ channel.

$$P=S(\textsf{Tree}(1, \textsf{Leaf}, \textsf{Leaf}), c, w) | R(c, w)$$

 With an addition to $P$, we can create a process $P'$ that reuses $c$ for transmitting an integer after transmitting the tree. 

$$P'=\left[S(\textsf{Tree}(1, \textsf{Leaf}, \textsf{Leaf}), c, w) | R(c, v)\right]|\left[v(c').c'(x).\nil|w(c'').\o{c''}\l1\r.\nil\right]$$

 The type of $c$ after expanding the recursive type is now $(T_C', \o{T_C'})$ where

\begin{lstlisting}[language=pi]
$T_C' = \lin \oplus\{$
    $\textsf{Leaf}: \lin \textsf{skip} $
    $\textsf{Node}: \lin \textsf{!Int};$
        $\mu z.\oplus\{\textsf{Leaf}: \lin \textsf{skip}\quad\textsf{Node}:\lin \textsf{!Int};z; z\};$
        $\mu z.\oplus\{\textsf{Leaf}:\lin  \textsf{skip}\quad\textsf{Node}:\lin \textsf{!Int};z;z\}$
$\};\lin !\textsf{Int}$
\end{lstlisting}

The last step is what motivates type equivalence; we would like to have a \textit{distributive law} that allows the output to be moved into the select type, as illustrated below.

\begin{lstlisting}[language=pi]
$T_C'' = \lin \oplus\{$
    $\textsf{Leaf}: \lin \textsf{skip};\lin !\textsf{Int} $
    $\textsf{Node}: \lin \textsf{!Int};$
        $\mu z.\oplus\{\textsf{Leaf}: \lin \textsf{skip}\quad\textsf{Node}:\lin \textsf{!Int};z;z\};$
        $\mu z.\oplus\{\textsf{Leaf}: \lin \textsf{skip}\quad\textsf{Node}:\lin \textsf{!Int};z;z\};$
        $\lin !\textsf{Int}$
$\}$
\end{lstlisting}
\end{example}

As the typing rules in our applied pi-calculus are defined on the type transitions instead of the structure of a type, the distributive law is not essential in our calculus, since we require that the type exhibit specific behaviour instead of having a specific structure, but the example shows how type equivalence can be used to tell if two types can be used interchangeably. In the next section we will introduce a method for checking if two types are equivalent, and we will return to Example \ref{ex:binarytree}, to check that it is indeed the case that $T_C''$ is equivalent to $T_C'$.

\subsection{Type bisimilarity}

As described in Section \ref{sec:typetransitions}, our types are highly reminiscent of BPA. For BPA expressions, bisimulation is used to prove that two processes exhibit the same behaviour. We now extend bisimulation to work on types as well. The definition follows from the definition of bisimulation in \cite[p. 37]{aceto2007reactive}.

\begin{definition} 
(Type bisimulation) A binary relation $\mathcal{R}$ between endpoint types is a type bisimulation iff whenever $T_{E_1} \mathcal{R}\ T_{E_2}$:
\begin{itemize}
    \item $Q(T_{E_1})=Q(T_{E_2})$
    \item $\forall \lambda $ if $T_{E_1} \trans{\lambda} T_{E_1}' $ then $\exists T_{E_2}'$ such that $ T_{E_2} \trans{\lambda} T_{E_2}'$ and $T_{E_1}' \mathcal{R}\ T_{E_2}'$
    \item $\forall \lambda $ if $T_{E_2} \trans{\lambda} T_{E_2}' $ then $\exists T_{E_1}'$ such that $ T_{E_1} \trans{\lambda} T_{E_1}'$ and $T_{E_1}' \mathcal{R}\ T_{E_2}'$
\end{itemize}
We write that $T_{E_1}\sim T_{E_2}$ if $T_{E_1} \mathcal{R} T_{E_2}$ for some type bisimulation and then say that $T_{E_1}$ and $T_{E_2}$ are \textit{type bisimilar}.
\end{definition}

Type bisimilar endpoint types will exhibit the same behaviour. In other words, for a type $T_{E}$, any type $T_E'$ where $T_E\sim T_E'$, $T_E'$ can be used instead of $T_E$, without introducing communication errors. 

\begin{example}
(A distributive law) Consider the types $T_C'$ and $T_C''$ from Example \ref{ex:binarytree}.

 We can use type bisimulation to show that these two types describe the same communication behaviour on a channel. To do so, we must provide a bisimulation that shows that $T_C'$ and $T_C''$ are type bisimilar.

Let $R$ be a relation over endpoint types. Let $R$ be the symmetric closure of 
\[\{(T_C',T_C''), ((\lin !\textsf{int};r;r);\lin !\textsf{int},\linebreak[0]\lin !\textsf{int};r;r;\lin !\textsf{int})\}\cup\{(T_E,T_E) | \forall T_E\in\mathcal{T}\}\]
where $r$ is the term $\mu z.\oplus\{\textsf{Leaf}: \lin \textsf{skip}\quad\textsf{Node}:\lin \textsf{!Int};z;z\}$. 


\end{example}

In fact, we can generalise this result and prove the distributive law for both select and branch.

\begin{lemma}
Let $\star$ be $\oplus$ or $\&$. Then $q\ \star\{l_i: T_{E_i}\}_{i\in I};T_E\sim q\ \star\{l_i: T_{E_i};T_E\}_{i\in I}$
\end{lemma}

\begin{proof}
Let $R$ be a relation between types. We must show that $R$ is a bisimulation of $q\ \star\{l_i: T_{E_i}\}_{i\in I};T_E$ and $q\ \star\{l_i: T_{E_i};T_E\}_{i\in I}$. Define $R$ as the symmetric closure of 
\[
\{(q\ \star\{l_i: T_{E_i}\}_{i\in I};T_E,q\ \star\{l_i: T_{E_i};T_E\}_{i\in I})\}\cup \{(T_E,T_E) | \forall T_E\in\mathcal{T}\}
\]

From the first requirement for a type bisimulation we have that $Q(q\ \star\{l_i: T_{E_i}\}_{i\in I};T_E) = Q(q\ \star\{l_i: T_{E_i};T_E\}_{i\in I}))$, this is trivially fulfilled as $q=q$.
By the \rt{Seq} rule we have that the transitions of a sequential compositions are those of the left-hand side of the operator. 
So the transitions of $q\ \star\{l_i: T_{E_i}\}_{i\in I};T_E$ are $q\ \star\{l_i: T_{E_i}\}_{i\in I};T_E\trans{\diamond l_i} T_{E_i};T_E$, where $\diamond \in \aset{\lhd,\rhd}$.
The available transitions for $q\ \star\{l_i: T_{E_i};T_E\}_{i\in I}$ are $q\ \star\{l_i: T_{E_i};T_E\}_{i\in I}\trans{\diamond l_i} T_{E_i};T_E$. Since $R$ is reflexive, we have that $(T_{E_i};T_E, T_{E_i};T_E)\in R$. So any transition taken by one of the types can be matched by the other to end up with syntactically equivalent types.
Since the types are syntactically equivalent, all further transitions can be matched by any of the two types, and all further type pairs will be in $R$. This means that $R$ is a type bisimulation, and that $q\ \star\{l_i: T_{E_i}\}_{i\in I};T_E \sim q\ \star\{l_i: T_{E_i};T_E\}_{i\in I}$.
\end{proof}

\subsection{A quotiented session type system}

We now define a quotiented session type system whose types are equivalence classes of session types from the the already existing type system. We define transitions between equivalence classes instead of endpoint types as follows.

\begin{definition} (Equivalence Classes)
Let $\vert T_E \vert$ be the equivalence class of $T_E$ given by:
\[
\vert T_E \vert = \{T_E'\in\mathcal{T} | T_E' \sim T_E\}
\]

We denote a transition between equivalence classes with action $\lambda$ as $\vert T_{E_1}\vert \trans{\lambda} \vert T_{E_2}\vert$.
\end{definition}

\begin{lemma} \label{lemma:typeclasses}
$\vert T_E \vert \trans{\lambda} \vert T_E' \vert $ iff $\forall T_{E_1} \in \vert T_E \vert \ \exists T_{E_1}' \in \vert T_E' \vert$ such that $T_{E_1} \trans{\lambda} T_{E_1}'$ 
\end{lemma}

\begin{proof}
By the properties of type bisimilarity that states that two bisimilar types will evolve to bisimilar types for all possible transitions.
\end{proof}

We use Lemma \ref{lemma:typeclasses} to define a new type system that is defined on equivalence classes of types from the previous type system. In Section \ref{sec:typetransitions} the transitions of endpoint types were of the form $T_E\trans\lambda T_E'$. In the new type system the transitions are of the form $\vert T_E \vert \trans\lambda \vert T_E'\vert$. The transition rules from Table \ref{tab:typetransitions} still apply to the new type system, but where all endpoint types $T_E$ have been replaced with $\vert T_E \vert$. For example, the \rt{Select} rule from Table \ref{tab:typetransitions} would in the new type system be \eqref{eq:selecteq}. We also expand the $Q$ function on equivalence classes to be the result of applying $Q$ to any witness of the equivalence class.
\begin{equation}
    \label{eq:selecteq}
    \rt{Select}\qquad\inference{q \sqsubseteq Q(\vert T_{E_k} \vert)}{\vert q\ \oplus\{l_i : \vert T_{E_i} \vert \}_{i\in I} \vert \trans{\triangleleft l_k} \vert T_{E_k}\vert }
\end{equation}
The typing rules in the new type system would also be the rules from Tables \ref{tab:extendedProcesses}, \ref{tab:termsrules} and \ref{tab:typesystem}, where endpoint types $T_E$ have been replaced with equivalence classes $\vert T_E \vert$. In \ref{eq:inputeq} the \rt{Input} rule from Table \ref{tab:typesystem} has been changed to fit the new type system.
\begin{equation}
    \label{eq:inputeq}
    \rt{Input} \qquad \inference{\Gamma_1 \vdash n:\vert T_E \vert & \vert T_E \vert \trans{?T} \vert T_E' \vert & \Gamma_2, x:T + n: \vert T_E' \vert \vdash P}{\Gamma_1 \circ \Gamma_2 \vdash n(x).P }
\end{equation}
The generic type system from \cite{Huttel2016} depends on the
transitions available to types. We have already shown that the
previous type system is an instance of generic type systems. From
Lemma \ref{lemma:typeclasses}, we can see that equivalence classes
have the exact same transitions as the old endpoint types had. From
these two results we can see that the type system defined on
equivalence classes is an instance of the generic type system as well,
allowing us to retain previously obtained results of fidelity and well
typed internal actions from Section \ref{sec:psicalculus}. In
conclusion, this gives us a type system where type equivalence is a
trivial property, since two endpoint types of the old type system
would be the same type in the new type system. 

\section{Conclusion}

In this paper we have considered a ``low-level'' applied pi-calculus
that allows composite terms to be built but only allows for passing
names and nullary function symbols. In this setting, we introduced a
type system for the applied pi-calculus based on the work on
context-free session types by Thiemann and Vasconcelos in
\cite{ThiemannVasco2016} and on the work on qualified session types by
Vasconcelos in \cite{Vasco2011}. The type system is a context-free
session type system with qualifiers. The type system is an instance of
the psi-calculus type system introduced in \cite{Huttel2016}, and this
allows us to establish a fidelity result about the type system that
ensures that a well typed process continues to be well typed, and
without communication errors, until termination.

The type system has a notion of type equivalence defined by
introducing type bisimilarity. Using it, we then get a
type system of equivalence classes, for which type equivalence is a
natural part of the type system itself.

The current focus is to deal with the decidability of type
equivalence. In \cite{ThiemannVasco2016} Thiemann and Vasconcelos show
the decidability of type equivalence using a transformation from their
types to guarded BPA expressions, for which bisimilarity is
decidable. As already discussed in this article, our types are very
close to BPA, in which case we should be able to achieve the same
results about decidability more directly. In \cite{Hirshfeld1994} an
algorithm is presented for deciding bisimilarity for normed context
free processes in polynomial time, and we conjecture that this
algorithm can be adapted to our setting for checking type
bisimilarity. Here, it would be important to find a characterization
of the class of applied pi processes that can be typed using normed
session types only. 

\bibliographystyle{eptcs}
\bibliography{sources.bib}
\end{document}